\documentclass[11pt]{article}
\pdfoutput=1
\usepackage[utf8]{inputenc}
\usepackage[T1]{fontenc}
\usepackage{fullpage}
\usepackage{authblk}
\usepackage{multirow}
\usepackage[dvipsnames,usenames]{color}
\usepackage[colorlinks=true,urlcolor=Blue,citecolor=Green,linkcolor=BrickRed]{hyperref}
\usepackage[usenames,dvipsnames]{xcolor}
\usepackage{amsthm,amsmath,amssymb,mathabx}
\usepackage{enumerate,comment}
\usepackage{url}
\usepackage{todonotes}
\usepackage[noadjust]{cite}

\newcommand{\Oh}{\mathcal{O}}
\newcommand{\leaves}{\mathsf{L}}
\newcommand{\id}{\mathsf{id}}
\newcommand{\finger}{\mathsf{finger}}
\newcommand{\lca}{\mathsf{lca\_ext}}
\newcommand{\level}{\mathsf{level}}
\newcommand{\head}{\mathsf{head}}

\newtheorem{lemma}{Lemma}
\newtheorem{theorem}{Theorem}

\newtheorem{proposition}{Proposition}

\makeatletter

\begin{document}

\title{A Faster Construction of Greedy Consensus Trees}

\author[1]{Paweł Gawrychowski}
\author[1]{Gad M. Landau}
\author[2]{Wing-Kin Sung}
\author[1]{Oren Weimann}
\affil[1]{University of Haifa, Israel}
\affil[2]{National University of Singapore, Singapore}
\date{}

\maketitle

\begin{abstract}
A consensus tree is a phylogenetic tree that captures the similarity between a set of conflicting phylogenetic trees.
The problem of computing a consensus tree is a major step in phylogenetic tree reconstruction.
It also finds applications in predicting a species tree from a set of gene trees.
This paper focuses on two of the most well-known and widely used consensus tree methods: the greedy consensus tree and the frequency difference consensus tree. Given $k$ conflicting trees each with $n$ leaves, the previous fastest algorithms for these problems were $\Oh(k n^2)$ for the greedy consensus tree [J. ACM 2016] and $\tilde \Oh(\min \{ k n^2, k^2n\})$ for the frequency difference consensus tree [ACM TCBB 2016]. We improve these running times to 
$\tilde \Oh(k n^{1.5})$ and $\tilde \Oh(k n)$ respectively. 
\end{abstract}

\thispagestyle{empty}
\clearpage
\setcounter{page}{1}

\section{Introduction}
A {\em phylogenetic tree} describes the evolutionary relationships among a set of $n$ species called taxa. It is an unordered rooted tree whose leaves represent the taxa and whose inner nodes represent their common ancestors. Each leaf has a distinct label from $[n]$. The inner nodes are unlabeled and have at least two children. 
 
Numerous phylogenetic trees, reconstructed from data sources like fossils or DNA sequences, have been published in the literature since the early 1860s. However, the phylogenetic trees obtained from different data sources or using different reconstruction methods result in conflicts (similar though not identical phylogenetic trees over the same set $[n]$ of leaf labels). The conflicts between phylogenetic trees are usually measured by their difference in \emph{signatures}: The signature of a phylogenetic tree $T$ is the set $\{ \leaves(u) : u \in T \}$ where $\leaves(u)$ denotes the set of labels of all leaves in the subtree rooted at node $u$ of $T$ (the set $\leaves(u)$ is sometimes called a {\em cluster}).
To deal with the conflicts between $k$ phylogenetic trees in a systematic manner, the concept of a {\em consensus tree} was invented. Informally, the consensus tree is a single phylogenetic tree that summarizes the branching structure (signatures) of all the conflicting trees. Consensus trees have been widely used in two applications:
\begin{enumerate}
\item {\bf Constructing a phylogenetic tree}: First, by sampling the dataset, we generate $k$ different datasets (for some constant $k$ that can be as large as $10,000$). Then, we reconstruct one phylogenetic tree for each dataset. Finally, we build the consensus tree of these $k$ trees.
\item {\bf Constructing a species tree}: First, a phylogenetic tree (called a gene tree) is reconstructed for each individual gene. Then, the species tree is created by building the consensus tree of all $k$ gene trees.
\end{enumerate}
%

Many different types of consensus trees have been proposed in the literature. For almost all of them,   optimal or near-optimal $\tilde \Oh(kn)$ time constructions are known. These include 
Adam's consensus tree~\cite{A72}, strict consensus tree~\cite{SR81}, loose consensus tree~\cite{B90,JanssonShenSung_consensus_tree_JACM2016}, majority-rule consensus tree~\cite{MM81,JanssonShenSung_consensus_tree_JACM2016}, majority-rule (+) consensus tree~\cite{Jansson_TCBB2016}, and  
asymmetric median consensus tree~\cite{PhillipsWarnow_asymmetric_median_tree_1996,Ramesh2017}\footnote{Constructing the asymmetric median consensus tree was proven to be NP-hard for $k>2$~\cite{PhillipsWarnow_asymmetric_median_tree_1996} and solvable in $\tilde \Oh(n)$ time for $k=2$~\cite{Ramesh2017}.}.  
Two of the most notable exceptions are the frequency difference consensus tree~\cite{software:GFN08} and the greedy consensus tree~\cite{chapter:Bryant03,software:Fel05} whose running time remains quadratic in either $k$ or $n$. In particular, the former can be constructed in $\tilde \Oh(\min \{ k n^2, k^2n\})$ time~\cite{Jansson_TCBB2016} and the later in $\Oh(k n^2)$ time~\cite{JanssonShenSung_consensus_tree_JACM2016}. For more details about different consensus trees and their advantages and disadvantages see the survey in~\cite{chapter:Bryant03}, Chapter~30 in~\cite{book:Fel04}, and Chapter~8.4 in~\cite{book:Sung_10}.

%

In this paper we propose novel algorithms for the frequency difference consensus tree problem and the greedy consensus tree problem. First, we present an $O(k n \log^2 n)$ time deterministic labeling method. The labelling method counts the frequency (number of occurrences) of every cluster $S$ in the input trees. 
Based on this labeling method, we obtain an $O(k n \log^2 n)$ time construction of the frequency difference consensus tree.
Then, for the greedy consensus tree, we present our main technical contribution: a method that uses micro-macro decomposition to
verify if a cluster $S$ is compatible with a tree $T$ in $\Oh(n^{0.5} \log n)$ time and, if so, modify $T$ to include $S$
in $\Oh(n^{0.5}\log n)$ amortized time.
Using this procedure, we obtain an $\Oh(k n^{1.5} \log n)$ time construction of the greedy consensus tree.

\paragraph{The frequency difference consensus tree.} The frequency $f(S)$ of a cluster $S$ (a set of labels of all leaves in some subtree) is the number of trees that contain $S$.
A cluster is said to be {\em compatible} with another cluster if they are either disjoint or one is included in the other. 
A {\em frequent} cluster is a cluster that occurs in more trees than any of the clusters that are incompatible with it. The frequency difference consensus tree is a tree whose signature is exactly all the frequent clusters.

The frequency difference consensus tree was initially proposed by Goloboff et al.~\cite{software:GFN08}, and its relationship with other consensus trees was studied in~\cite{DF-BMP_10}.  
In particular, it can be seen as a refinement of the majority-rule consensus tree~\cite{MM81,JanssonShenSung_consensus_tree_JACM2016}.
Moreover, it is known to give less noisy branches than the greedy consensus tree defined below.
Steel and Velasco~\cite{SV_14} concluded that ``the frequency difference method is worthy of more widespread usage and serious study''.
A naive construction of the frequency difference consensus tree  takes $\Oh(k^2 n^2)$ time.
The free software TNT~\cite{software:GFN08} has implemented a heuristics method to construct it more efficiently.
However, its time complexity remains unknown.

Recently, Jansson et al.~\cite{Jansson_TCBB2016} presented an $\Oh(\min \{ kn^2, k^2n + kn\log^2 n \})$ time construction (implemented in the FACT software package~\cite{FACT}).
Their algorithm first computes the frequency $f(S)$ of every cluster $S$ with non-zero frequency. This is done in total $\Oh(\min \{ kn^2, k^2n \})$ time. They then show that given these computed frequencies, the frequency difference consensus tree can be computed in additional $\Oh(k n \log^2 n)$ time. In Section~\ref{sec:Identifiers} we show how to compute all frequencies in total $\Oh(k n \log^2 n)$ time leading to the following theorem:
%
%

\begin{theorem}\label{thm:frequency}
The frequency difference consensus tree of $k$ phylogenetic trees $T_1,T_2,\ldots,T_k$ on the same set of leaves $[n]$  can be computed in $\Oh(k n \log^2 n)$ time.
\end{theorem}

To prove the above theorem, we first develop an $\Oh(k n\log^{2}n)$ time algorithm for
assigning a number $\id(u)\in [k n]$ to every $u\in T_{i}$ such that $\id(u)=\id(u')$ iff
$\leaves(u)=\leaves(u')$. 
With these numbers in hand, we can then compute the frequencies of all clusters in $\Oh(k n)$ time using counting sort (since there are only $k n$ clusters with non-zero frequencies, and each was assigned an integer bounded by $kn$). Notice that this also generates a sorted list of all clusters with non-zero frequencies.

\paragraph{The greedy consensus tree.} 
We say that a given collection $\mathcal{C}$ of subsets of $[n]$ is \emph{consistent} if there exists a phylogenetic tree $T$ such that the signature of $T$ is exactly $\mathcal{C}$. 
The greedy consensus tree is defined by the following procedure: We begin with an initially empty $\mathcal{C}$ and then consider all clusters $S$ in decreasing order of their frequencies. In this order, for every $S$, we check if $\mathcal{C}\cup \{S\}$ is consistent, and if so we add $S$ to $\mathcal{C}$. 

The greedy consensus tree is one of the most well-known consensus trees. It has been used in numerous papers such as~\cite{app4,Nature,Jarvis1320,MirarabBayzid,BayzidW13,LeonidasStamatakis,AlexandrosJacquesRenner,PeaseHaakHahn,BayzidMirarabWarnow,905976,LIU2009320,SmithBraun,YangW11} to name a few. For example, in a recent landmark paper in Nature~\cite{Nature}, it was used to construct the species tree from 1000 gene trees of yeast genomes, and in~\cite{app4} it was asserted that {\em ``The greedy consensus tree offers some robustness to gene-tree discordance that may cause other methods to fail to recover the species tree. In addition, the greedy consensus method outperformed our other methods for branch lengths outside the too-greedy zone.''}. 

The greedy consensus tree is an extension of the majority-rule consensus tree, and is sometimes called the extended majority-rule consensus (eMRC) tree.
It is implemented in popular phylogenetics software packages like PHYLIP~\cite{software:Fel05}, PAUP*~\cite{software:Swo03},  MrBayes~\cite{software:RH03}, and RAxML~\cite{Stamatakis14}. 
A naive construction of the greedy consensus tree requires $\Oh(k n^3)$ time~\cite{chapter:Bryant03}. 
To speed this up, these software packages use some forms of randomization methods.
For example, PHYLIP uses hashing to improve the running time.
Even with randomization, the time complexities of these solutions are not known.
Recently, Jansson et al.~\cite{JanssonShenSung_consensus_tree_JACM2016} gave the best known provable construction with an $\Oh(k n^2)$ deterministic running time (their implementation is also part of the FACT package).
In Section~\ref{sec:greedy} we present our main contribution, a deterministic $\tilde \Oh(k n^{1.5})$ construction as stated by the following theorem:
 
\begin{theorem}
The  greedy consensus tree of $k$ phylogenetic trees $T_1,T_2,\ldots,T_k$ on the same set of leaves $[n]$ can be computed in $\Oh(k n^{1.5} \log n)$ time.
\end{theorem}

To prove the above theorem, we develop a generic procedure that takes any ordered list of clusters $S_{1},S_{2},\ldots,S_{\ell}\subseteq [n]$ and tries adding them one-by-one to the current solution $\mathcal{C}$. 
We assume that every cluster $S_{i}$ is specified by providing a tree $T_{i}$ and a node $u_{i}\in T_{i}$ such that $S_{i}=\leaves(u_{i})$. Our procedure requires  $\Oh(n^{0.5} \log n)$ time per cluster (to add this cluster to $\mathcal{C}$ or  assert that it cannot be added) and needs not to assume anything about the order of the clusters. In particular, it does not rely on the clusters being sorted by frequencies.

\section{Computing the Identifiers}\label{sec:Identifiers}

We process the nodes of every $T_{i}$ in the bottom-up order. For every node $u\in T_{i}$, we compute the identifier $\id(u)$ by updating the following structure called the {\em dynamic set equality structure}: 

\begin{lemma}
\label{lem:dynamicset}
There exists a dynamic set equality structure that supports: (1) create a new empty structure in constant time, (2) add $x\in [n]$ to the current set in $\Oh(\log^{2}n)$ time, (3) return the identifier of the current set in constant time, and (4) list all $\ell$ elements of the current set in $\Oh(\ell)$ time.  The structure ensures that the identifiers are bounded by the total number of update operations performed so far, and that two sets are equal iff their identifiers are equal.
\end{lemma}

\begin{proof}
To allow for listing all elements of the current set $S$, we store them in a list. Before adding the new
element $x$ to the list, we need to check if $x\in S$. This will be done using the representation described
below.

Conceptually, we work with a complete binary tree $B$ on $n$ leaves labelled with $0,1,\ldots,n-1$
when read from left to right (without losing generality, $n=2^{k}$), where every node $u$ corresponds
to a set $D(u)\subseteq [n]$ defined by the leaves in its subtree (note that $D(u)=\{i,i+1,\ldots,j\}$, where $0 \leq i \leq j < n$).
Now, any set $S$ is associated with a binary tree $B$, where we write $1$ in a leaf if the corresponding element belongs to $S$
and $0$ otherwise. Then, for every node we define its characteristic vector by writing down the
values written in the leaves of its subtree in the natural order (from left to right). Clearly, the vector
of an inner node is obtained by concatenating the vector of its children. We want to maintain
identifiers of all nodes, so that the identifiers of two nodes are equal iff their characteristic vectors
are identical. If we can keep the identifiers small, then the identifier of the current set can be
computed as the identifiers of the root of $B$. 

Assume that we have already computed the identifiers of all nodes in $B$ and now want to
add $x$ to $S$. This changes the value in the leaf $u$ corresponding to $x$ and, consequently,
the characteristic vectors of all ancestors of $u$. However, it does not change the characteristic
vectors of any other node. Therefore, we traverse the ancestors of $u$ starting from $u$ and
recompute their identifiers. Let $v$ be the current node. If we have never seen the characteristic
vector of $v$ before, we can set the identifier of $v$ to be the largest already used identifier
plus one. Otherwise, we have to set the identifier of $v$ to be the same as the one previously
used for a node with such a characteristic vector. As mentioned above, the characteristic vector
of an inner node $v$ is the concatenation of the characteristic vectors of its children $v_{\ell}$
and $v_{r}$. We maintain a dictionary mapping a pair consisting of the identifier of $v_{\ell}$
and the identifier of $v_{r}$ to the identifier of $v$. The dictionary is global, that is, shared by all
instances of the structure.  Then, assuming that we have already computed the up-to-date identifiers
of $v_{\ell}$ and $v_{r}$, we only need to query the dictionary to check if the identifier of $v$ should be
set to the largest already used identifier plus one (which is exactly when the dictionary does not contain
the corresponding pair) or retrieve the appropriate identifier. Therefore, adding $x$ to $B$
reduces to $\log n$ queries to the dictionary. By implementing the dictionary with balanced search trees,
we therefore obtain the claimed $\Oh(\log^{2}n)$ time for adding an element.

We are not completely done yet, because creating a new complete binary tree $B$ takes $\Oh(n)$
time and therefore the initialization time is not constant yet. However, we can observe that
it does not make sense to explicitly maintain a node $u$ of $B$ such that $S\cap D(u)=\emptyset$,
because we can assume that the identifier of such an $u$ is 0. In other words, we can maintain
only the part of $B$ induced by the leaves corresponding to $S$. Adding an element $x\in S$
is implemented as above, except that we might need to create (at most $\Oh(\log n)$) new nodes on the 
leaf-to-root path corresponding to $x$ (if such a leaf already exists, we terminate the procedure
as $x\in S$ already) and then recompute the identifiers on the whole path as described above.
\end{proof}

Armed with Lemma~\ref{lem:dynamicset}, we process every $T_{i}$ bottom-up. Consider an inner node $v\in T_{i}$
and let $v_{1},v_{2},\ldots,v_{d}$ be its children ordered so that $|\leaves(v_{1})|=\max_{j}|\leaves(v_{j})|$,
that is, the subtree rooted at $v_{1}$ is the largest. Assuming that we have already stored every $\leaves(v_{j})$
in a dynamic set equality structure, we construct a dynamic set equality structure storing $\leaves(v)$
by simply inserting all elements of $\leaves(v_{2})\cup\leaves(v_{3})\cup\cdots \cup \leaves(v_{d})$ into the
structure of $\leaves(v_{1})$. This takes $\Oh(\log^{2} n)$ time per element. Then, we set
$\id(u)$ to be the identifier of the obtained structure. By a standard argument (heavy path decomposition), every leaf of $T_{i}$
is inserted into at most $\log n$ structures and therefore the whole $T_{i}$ is processed in $\Oh(n\log^{3}n)$
time. This gives us the claimed $\Oh(k n\log^{3}n)$ total time.

We now proceed with a faster $\Oh(k n\log^{2}n)$ total time solution. While this is irrelevant
for our $\Oh(k n^{1.5}\log n)$ time construction of the greedy consensus tree,
it implies a better complexity for constructing the frequency difference consensus tree.

We start with a high-level intuition. Lemma~\ref{lem:dynamicset} is, in a sense, more than we
need, as it is not completely clear that we need to immediately compute the identifier of the
current set. Indeed, applying heavy path decomposition we can partially delay computing the identifiers
by proceeding in $\Oh(\log n)$ phases. In each phase, we can then replace the dynamic dictionary
 used to store the mapping with a radix sort. Intuitively, this shaves one log from the time
complexity. We proceed with a detailed explanation.

\begin{theorem}
\label{thm:fastersorting}
The numbers $\id(u)$ can be found for all nodes of the $k$ phylogenetic trees $T_{1},T_{2},\ldots,T_{k}$
in $\Oh(k n\log^{2} n)$ total time.
\end{theorem}

\begin{proof}
For a node $v\in T_{i}$, define its level $\level(v)$ to be $\ell$, such that $2^{\ell}\leq |\leaves(v)|<2^{\ell+1}$.
Thus, the levels are between $0$ and $\log n$, level of a node is at least as large as the levels of its
children, and a node on level $\ell$ has at most one child on the same level. We work in phases
$\ell=0,1,\ldots,\log n$. In phase $\ell$, we assume that the numbers $\id(v)$ are already known for
all nodes $v$, such that $\level(v) < \ell$, and want to assign these numbers to all nodes $v$, such
that $\level(v)=\ell$. We will show how to achieve this in $\Oh(k n\log n)$ time, thus proving the
theorem.

Consider all nodes $v$, such that $\level(v)=\ell$. Because every such $v$ has at most one child at the
same level, all level-$\ell$ nodes in $T_{i}$ can be partitioned into maximal paths of the form $p=v_{1}-v_{2}-\ldots -v_{s}$,
where the level of the parent of $v_{1}$ is larger than $\ell$ (or $v_{1}$ is the root of $T_{i}$), and the
levels of all children of $v_{j}$ (except for $v_{j+1}$, if defined) are smaller than $\ell$. 
$v_{1}$ is called the head of $p$ and denoted $\head(p)$. Now, our goal is
to find $\id(v_{j})$ with the required properties for every $j=1,2,\ldots,s$. We will actually
achieve a bit more. The sets $\leaves(\head(p))$ are disjoint in every tree $T_{i}$, and thus
we can define, for every $i$, a partition $\mathcal{P}_{i}=\{P_{i}(1),P_{i}(2),\ldots,P_{i}(t_{i})\}$
of the set of leaves $[n]$, where every $P_{i}(z)$ corresponds to a level-$\ell$ path
$p=v_{1}-v_{2}-\ldots - v_{s}$ in $T_{i}$, such that $\leaves(\head(p))=P_{i}(z)$. 
The elements of $P_{i}(z)$ are then ordered, and we think that $P_{i}(z)$ is a sequence
of length $|P_{i}(z)|$. The ordering is chosen so that, for every $j=1,2,\ldots,s$, the set
$\leaves(v_{j})$ corresponds to some prefix of $P_{i}(z)$. $P_{i}(z)[1..r]$ denotes the prefix
of $P_{i}(z)$ of length $r$. We will assign identifiers to all such prefixes $P_{i}(z)[1..r]$,
for every $i=1,2,\ldots,k$, $z=1,2,\ldots,t_{i}$ and $r=1,2,\ldots,|P_{i}(z)|$, with the property that the identifiers
of two prefixes are equal iff the sets of leaves appearing in both of them are equal.
Then, we can extract the required $\id(v_{j})$ in constant time each by taking the
identifiers of some $P_{i}(z)[1..r]$.

Recall that in the slower solution we worked with a complete binary tree $B$ on $n$ leaves.
For every set $S$ in the collection and every $u\in B$, we computed an identifier of the set
$S \cap D(u)$. This was possible, because if $u_{\ell}$ and $u_{r}$ are the left and the right
child of $u$, respectively, then the identifier of $S \cap D(u)$ can be found using the identifiers
of $S \cap D(u_{\ell})$ and $S \cap D(u_{r})$. We need to show that retrieving these identifiers
can be batched.

Fix a node $u\in B$ and, for every $i=1,2,\ldots,k$ and $z=1,2,\ldots,t_{i}$, consider
all prefixes $P_{i}(z)[1..r]$ for $r=1,2,\ldots,|P_{i}(j)|$. We create a {\em version} of $u$ for every
such prefix. The version corresponds to the set containing all elements of $D(u)$ occurring
in the prefix $P_{i}(z)[1..r]$. We want to assign identifiers to all versions of $u$. 
First, observe that we only have to create a new version if $P_{i}(z)[r] \in D(u)$, as otherwise
the set is the same as for $r-1$. Thus, the total number of required versions, when summed over
all nodes $u\in B$ on the same depth in $B$, is only $k n$, as a leaf of $T_{i}$
creates exactly new version for some $u$. For every node $u\in B$, we will
store a list of all its versions. A version consists of its identifier (such that the identifier
of two versions is the same iff the corresponding sets are equal) together with
the indices $i$, $z$ and $r$. We describe how to create such a list for every node $u\in B$
at the same depth $d$ given the lists for all nodes at depth $d+1$ next.

Let $u_{1}$ and $u_{2}$ be the left and the right child of $u\in B$, respectively.
Then, we need to create a new version of $u$ for every new version of $u_{1}$ and
every new version of $u_{2}$, because for the set corresponding to $u$ to change
either the set corresponding to $u_{1}$ or the set corresponding to $u_{2}$ must
change, and every change is adding one new element. Fix $i$ and $z$ and consider
all versions of $u_{1}$ corresponding to $i$ and $z$ sorted according to $r$.
Let the sorted list of their $r$'s be $a_{1} < a_{2} < \ldots$. Similarly, consider
all versions of $u_{r}$ corresponding to $i$ and $z$ sorted according to $r$,
and let the sorted list of their $r$'s be $b_{1} < b_{2} < \ldots$. For every $x\in \{a_{1},a_{2},\ldots\}\cup \{b_{1},b_{2},\ldots\}$,
we create a new version of $u$ corresponding to $i$, $z$, and $r$ equal to $x$.
This is done by retrieving the version of $u_{1}$ with $r$ equal to $a_{p}$, such that
$a_{p}\leq x$ and $p$ is maximized, and the version of $u_{2}$ with $r$ equal to $b_{q}$,
such that $b_{q}\leq x$ and $q$ is maximized. Then, the identifier of the new version
of $u$ can be constructed from the pair consisting of the identifiers of these versions of
$u_{1}$ and $u_{2}$ (this is essentially the same reasoning as in the slower solution).
We could now use a dictionary to map these pairs to identifiers. However, we can also
observe that, in fact, we have reduced finding the identifiers of all versions of
all nodes $u\in B$ at the same depth $d$ to identifying duplicates on a list of $k n$
pairs of numbers from $[k n]$. This can be done by radix sorting all pairs
in linear time (more precisely, $\Oh(k n)$ time and $\Oh(k n)$ space),
and then sweeping through the sorted list while assigning the identifiers.
This takes only $\Oh(k n)$ time for every depth $d$, so $\Oh(k n\log n)$
for every level as claimed.
\end{proof}

The proof of Theorem~\ref{thm:frequency} follows immediately from Theorem~\ref{thm:fastersorting}.

\section{Simulating the Greedy Algorithm}\label{sec:greedy}

We consider $k$ trees $T_1, \ldots, T_k$ on the same set of leaves $[n]$, and 
assume that every node $u$ has an identifier $\id(u)$ such that $\id(u)=\id(u')$ iff $\leaves(u)=\leaves(u')$.
We next develop a general method for maintaining a 
solution $\mathcal{C}$ (i.e., a set of compatible identifiers) so that, given any node $u\in T_i$, we are able to efficiently check if
$\leaves(u)$ is compatible with $\mathcal{C}$, meaning that $\mathcal{C}\cup\leaves(u)$
is consistent, and if so add $\leaves(u)$ to $\mathcal{C}$. Our method does not rely on the order in which the sets arrive and in particular can be used to run the greedy algorithm. 

We represent $\mathcal{C}$ with a phylogenetic tree $T_c$ such that $\mathcal{C}=\{\leaves(u): u\in T_c\}$.
$T_c$ is called the \emph{current consensus tree}. By Lemma 2.2 of~\cite{JanssonShenSung_consensus_tree_JACM2016}, 
$S$ is compatible with $\mathcal{C}$
iff there exists a node $v\in T$ such that for every child $v'$ of $v$ either
$\leaves(v')\cap S = \emptyset$ or $\leaves(v')\subseteq S$. Also, adding $\leaves(u)$ to
$\mathcal{C}$ can be done by creating a new child $w$ of $v$ and reconnecting every original
child $v'$ of $v$ such that $\leaves(v')\subseteq S$ to the new $w$.
This is illustrated in Figure~\ref{fig:adding}.

\begin{figure}[h]
\begin{center}
\includegraphics{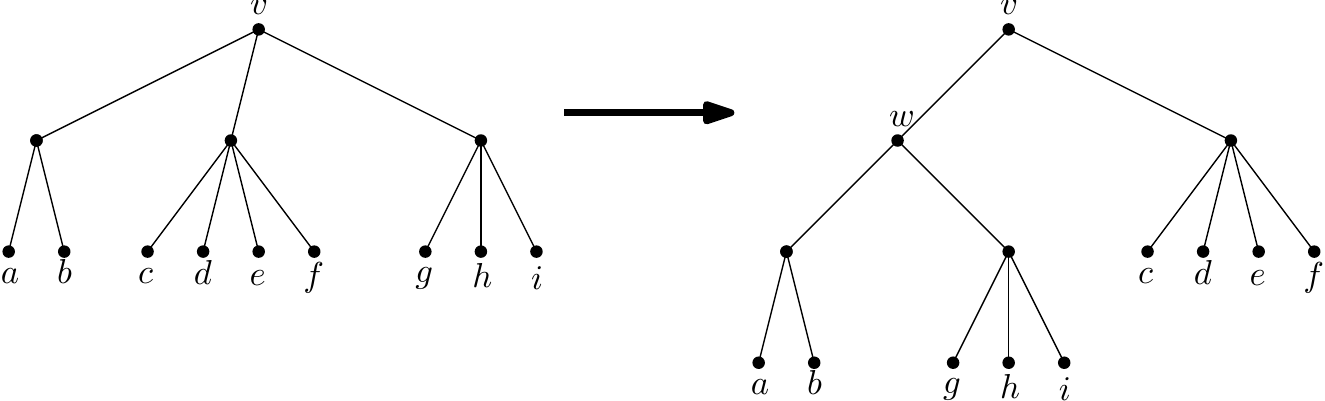}
\end{center}
\caption{Adding $\{a,b,g,h,i\}$ to $S$.\label{fig:adding}}
\end{figure}

Initially, $T_c$ consists only of $n$ leaves attached to the common root (which corresponds to
$\mathcal{C}=\{\{x\} : x\in [n]\}$). Our goal is to to maintain some additional information
so that given any node $u\in T_i$, we can check if $\leaves(u)$ is compatible with $\mathcal{C}$
in $\Oh(n^{0.5}\log n)$ time. After adding $\leaves(u)$ to $\mathcal{C}$ the information will be
updated in amortized $\Oh(k n^{0.5}\log n)$ time. To explain the intuition, we first
show how to check if $\leaves(u)$ is compatible with $\mathcal{C}$ in roughly $\Oh(|\leaves(u)|)$
time.

\begin{figure}[h!]
\begin{center}
\includegraphics{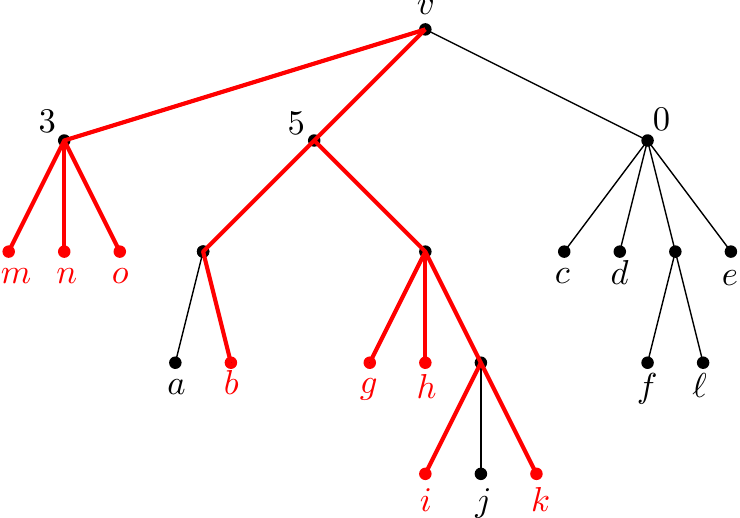}
\end{center}
\caption{Checking if $S=\{m,n,o,b,g,h,i,k\}$ is compatible with $T_{c}$. Leaves corresponding to the
elements of $S$ are shown in red and their lca is $v$. $S$ is not compatible with $T_{c}$ because the counter of the middle child of $v$ is equal to 5 yet there are 7 leaves in its subtree.\label{fig:notconsistent}
}
\end{figure}

Let $\leaves(u)=\{\ell_{1},\ell_{2},\ldots,\ell_{s}\}$ and let $u_{i}$ be the leaf of $T_{c}$ labelled
with $\ell_{i}$. Then, $v$ must be an ancestor of every $u_{i}$. We claim that, in fact, $v$ should
be chosen as the lowest common ancestor of $u_{1},u_{2},\ldots,u_{s}$, because if all $u_{i}$'s
are in the same subtree rooted at a child $v'$ of $v$ then we can as well replace $v$ with $v'$.
So, we can find $v$ by asking $s-1$ lca queries: we start with $u_{1}$ and then iteratively jump
to the lca of the current node and $u_{i}$. Assuming that we represent $T_{c}$ in such a way
that an lca query can be answered efficiently, this takes roughly $\Oh(s)$ time. 
Then, we need to decide if for every child $v'$ of $v$ it holds that $\leaves(v')\subseteq\leaves(u)$
or $\leaves(v')\cap\leaves(u)=\emptyset$. This can be done by computing, for every such $v'$,
how many $u_i$'s belong to the subtree rooted at $v'$, and then checking if this number
is either 0 or $|\leaves(v')|$. To compute these numbers, we maintain a counter for every $v'$.
Then, for every $u_{i}$ we retrieve the child $v'$ of $v$ such that $u_{i}$ belongs to the subtree
rooted at $v'$ and increase the counter of $v'$. Assuming that we represent $T_{c}$ so that
such $v'$ can be retrieved efficiently, this again takes roughly $\Oh(s)$ time.
Finally, we iterate over all $u_{i}$ again, retrieve the corresponding $v'$ and check if
its counter is equal to $|\leaves(v')|$ (so our representation of $T_{c}$ should also allow
retrieving the number of leaves in a subtree). If not, then $\leaves(u)$ is not compatible
with $\mathcal{C}$, see Figure~\ref{fig:notconsistent}. Otherwise, we create the new node $w$
and reconnect to $w$ all children $v'$ of $v$, such that the counter of $v'$ is equal to $|\leaves(v')|$.

We would like to avoid explicitly iterating over all elements of $\leaves(u)$. This will be done
by maintaining some additional information, so that we only have to iterate over up to $n^{0.5}$
elements. To explain what is the additional information we need the (standard) notion of a
\emph{micro-macro decomposition}. Let $b$ be a parameter and consider a binary tree
on $n$ nodes. We want to partition it into $\Oh(n/b)$ node-disjoint subtrees called
\emph{micro trees}. Each micro tree is of size at most $b$ and contains at most two
\emph{boundary nodes} that are adjacent to nodes in other micro trees. One of these
boundary nodes, called the top boundary node, is the root of the whole micro tree,
and the other is called the bottom boundary node. Such a partition is always possible
and can be found in $\Oh(n)$ time.

We binarize every $T_i$ to obtain $T'_i$. Then, we find a micro-macro decomposition
of $T'_i$ with $b=n^{0.5}$. By properties of the decomposition we have the following:

\begin{proposition}
For any $u\in T_i$ such that $|\leaves(u)|>n^{0.5}$, there exists a boundary node $v\in T'_i$ such that $\leaves(u)$ can be
obtained by adding at most $n^{0.5}$ elements to $\leaves(v)$. Furthermore, $v$ and these up to
$n^{0.5}$ elements can be retrieved in $\Oh(n^{0.5})$ time after $\Oh(n)$ preprocessing.
\end{proposition}

\begin{figure}[h]
\begin{center}
\includegraphics{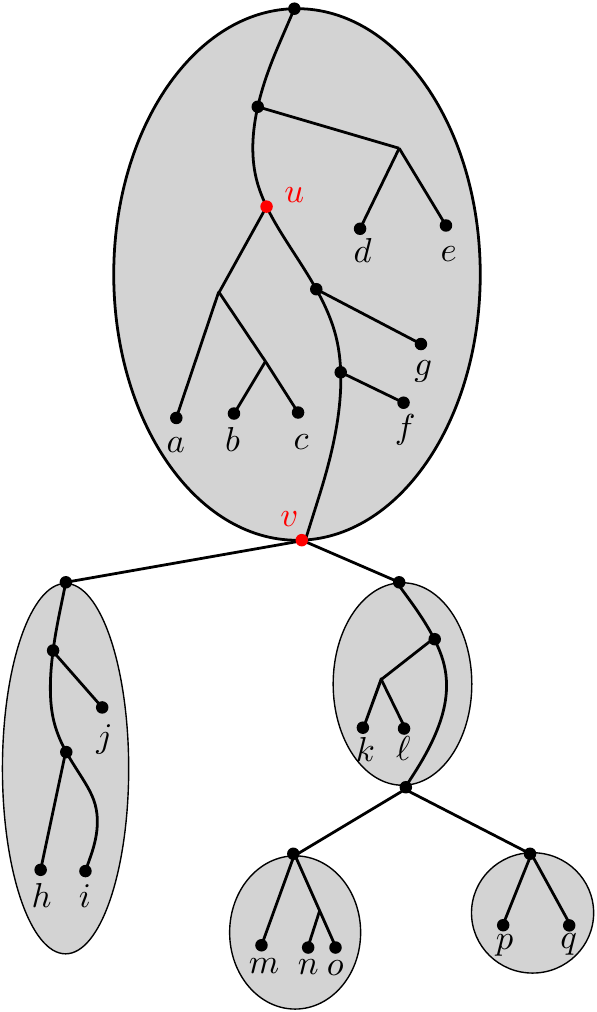}
\end{center}
\caption{A schematic illustration of the micro-macro decomposition. $v$ is a boundary node
and $\leaves(v)=\{h,i,j,k,\ell,m,n,o,p,q\}$. Then, $\leaves(u)=\{a,b,c,f,g,h,i,j,k,\ell,m,n,o,p,q\}$ 
so $\leaves(u)=\leaves(v)\cup\{a,b,c,f,g\}$.}
\end{figure}

The total number of boundary nodes is only $\Oh(k n^{0.5})$. For each such boundary
node $u$, we maintain a pointer to a node $\finger(u)\in T_c$ called the \emph{finger} of $u$.
$\finger(u)$ is a node $v\in T_c$ such that $\leaves(u)\subseteq \leaves(v)$
but, for every child $v_i$ of $v$,  $\leaves(u)\not\subseteq \leaves(v_i)$.

\begin{proposition}
The node $\finger(u)$ is the lowest common ancestor in $T_c$ of all leaves with labels belonging
to $\leaves(u)$.
\end{proposition}

Additionally, the children of $\finger(u)$
are partitioned into three groups: (1) $v_i$ such that $\leaves(v_i)\subseteq\leaves(v)$,
(2) $v_i$ such that $\leaves(v_i)\cap\leaves(v)=\emptyset$, and (3) the rest. We call them
full, empty, and mixed, respectively (with respect to $u$). For each group we maintain
a list storing all nodes in the group, every node knows its group, and the group knows it size.
Additionally, every group knows the total number of leaves in all subtrees rooted at
its nodes.

We also need to augment the representation $T_c$ to allow for efficient \emph{extended lca
queries}. The lowest common ancestor (lca) of $u$ and $v$ is the leafmost node $w$ that is an
ancestor of both $u$ and $v$. An extended lca query, denoted $\lca(u,v)$, returns the first
edge on the path from the lca of $u$ and $v$ to $u$, and -1 if $u$ is an ancestor of $v$. For example, in Figure~\ref{fig:notconsistent}, $\lca(v,k)=-1$ whereas $\lca(n,k)$ is the edge between $v$ and its leftmost child. 

\begin{lemma}
\label{lem:lca}
We can maintain a collection of rooted trees under: (1) create a new tree consisting of a single
node, (2) make the root of one tree a child of a node in another tree, (3) delete an edge
from a node to its parent, (4) count leaves in the tree containing a given node,
and (5) extended lca queries, all in $\Oh(\log n)$ amortized time, where $n$
is the total size of all trees in the collection.
\end{lemma}

\begin{proof}
We apply the link/cut trees of Sleator and Tarjan~\cite{SleatorT83} to maintain the collection.
This immediately gives us the first three operations. 
To implement computing the size and $\lca(u,v)$ queries we need to explain the internals of link/cut
trees. Each tree is partitioned into node-disjoint paths consisting of \emph{preferred edges}. Each node has
at most one such edge leading to its preferred child. For each maximal path consisting of preferred
edges, called a preferred path, we store
its nodes in a splay tree, where the left-to-right order on the nodes of the splay tree corresponds
to the top-bottom order on the nodes in the rooted tree. Each such splay tree stores a pointer to
the topmost node of its preferred path. Additionally, each node of the tree stores a pointer to its current
parent. All operations on a link/cut tree use the access procedure. Its goal is to change the
preferred edges so that there is a preferred path starting at the root and
ending at $v$. This is done by first shortening the preferred path containing $v$ so that
it ends at $v$. Then, we iteratively jump to the topmost node $u$ of the current preferred path
and make $u$ the preferred child of its parent. Whenever the preferred child of a node
changes, we need to update the splay tree representing the nodes of the preferred path.
Even though the number of jumps might be $\Omega(n)$, it can be shown that all these
updates take $\Oh(\log n)$ amortized time.

Now we can explain how to implement $\lca(u,v)$. First, we access node $v$. This gives
us a preferred path starting at the root and ending at $v$. Second, we access node $u$ while
keeping track of the topmost nodes of the visited preferred paths. If $u$ is on the same preferred
path as $v$, then $u$ is an ancestor of $v$. Otherwise, let $p$ be the preferred path visited just before
reaching the preferred path starting at the root of the whole tree. Then the topmost node of $p$
(before changing the preferred child of its parent) should be returned as $\lca(u,v)$. Thus,
the complexity of $\lca(u,v)$ is the same as the complexity of access.

To compute the size of a tree, we augment the splay trees. Every node of a preferred path
stores the total number of leaves in all subtrees attached to it through non-preferred edges
(plus one if the node itself is a leaf). Additionally, every node of a splay tree stores the sum
of the numbers stored in its subtree, or in other words the total number of leaves in all subtrees
attached to its corresponding contiguous fragment of the preferred path through non-preferred
edges. The sums stored at the nodes of the splay tree are easily maintained during rotations.
We also need to update the total number of leaves after making a preferred edge non-preferred
or vice versa. This is easily done by accessing the sum stored at the root of the splay tree.
To access the number the leaves in the tree containing $v$, we need to access $v$. This makes
all of $v$'s children non-preferred and makes $v$ the root of its splay tree. Hence, the number stored at $v$ is the total number of leaves in the tree containing $v$.
\end{proof}

We next show how to efficiently check for any $u$ if $\leaves(u)$ is compatible with $\mathcal{C}$. By the following lemma, this can be done in $\Oh(n^{0.5}\log n)$ time, assuming we have stored  the required additional information. Recall that this additional information includes:
\begin{enumerate}
\item The tree $T_c$ maintained using Lemma~\ref{lem:lca}.
\item For every boundary node $w$,
we store $\finger(w)$. 
\item For every boundary node $w$, we store three lists containing the full, the mixed, and the empty children of $w$ respectively. Each list also stores the total number of leaves in all subtrees rooted at its nodes.
\end{enumerate}

\begin{lemma}
\label{lem:query}
Assuming access to the above additional information, given any node $u\in T_i$
we can check if $\leaves(u)$ is compatible with $\mathcal{C}$ in $\Oh(n^{0.5}\log n)$ time.
\end{lemma}

\begin{proof}
By Lemma 2.2 of~\cite{JanssonShenSung_consensus_tree_JACM2016}, to check if $\leaves(u)$ is compatible with $\mathcal{C}$ we need
to check if there exists a node $v$ such that for every child $v'$ of $v$ either
$\leaves(v')\cap\leaves(u)=\emptyset$
or $\leaves(v')\subseteq \leaves(u)$. First, observe that $v$ can be chosen as the lowest
common ancestor of all leaves with labels belonging to $\leaves(u)$.
By properties of the micro-macro decomposition, we can retrieve a boundary node
$w$ and a set $S$ of up to $n^{0.5}$ labels such that $\leaves(u)=\leaves(w)\cup S$
(if $|\leaves(u)|<n^{0.5}$, there is no $w$).
Then, the lowest common ancestor of all leaves with labels belonging to $\leaves(u)$
is the lowest common ancestor of $\finger(w)$ and all leaves with labels belonging
to $S$. Therefore, $v$ can be found with $|S|$ lca queries in $\Oh(n^{0.5}\log n)$ time.
Second, to check if $\leaves(v_{i})\cap\leaves(u)=\emptyset$ or $\leaves(v_{i})\subseteq \leaves(u)$
for every child $v_i$ of $v$ we distinguish two cases:

If $v$ is a proper ancestor of $\finger(w)$ we can calculate $|\leaves(v_i)\cap\leaves(u)|$
for every $v_i$ in $\Oh(|S|\log n)=\Oh(n^{0.5}\log n)$ time as follows. Every edge has its
associated counter. We assume that all counters are set to zero before starting the procedure
and will make sure that they are cleared at the end. First, we use an $\lca(w,v)$ query to access the
edge leading to the subtree containing $w$ and set its counter to $|\leaves(w)|$. Then,
we iterate over all $\ell\in S$, retrieve the leaf $u$ of $T_{c}$ labelled with $\ell$,
and use an $\lca(u,v)$ query to access the edge leading to the subtree of $v$ containing $u$
and increase its counter by one. Additionally, whenever we access
an edge for the first time (in this particular query), we add it to a temporary list $Q$. After
having processed all $\ell\in S$, we iterate over $(v,v_i)\in Q$ and check if the counter of $(v,v_i)$
is equal to the number of leaves in the subtree rooted at $v_i$ (which requires retrieving the
number of leaves). If this condition holds for every $(v,v_i)\in Q$ then $\leaves(u)$ is compatible
with $\mathcal{C}$ and furthermore, the nodes $v_{i}$ such that $(v,v_i)\in Q$ are exactly
the ones that should be reconnected. Finally, we iterate over the edges in $Q$ again and reset
their counters.

If $v=\finger(w)$ the situation is a bit more complicated because we might not have enough
time to explicitly iterate over all children of $v$ that should be reconnected. Nevertheless,
we can use a very similar method. Every edge has its associated counter (again, we assume
that the counter are set to zero before starting the procedure and will make sure that they
are cleared at the end). We also need a global counter $g$, that is set to the total number
of leaves in all subtrees rooted at either full or mixed children of $v$ decreased by $|\leaves(w)|$.
$g$ can be initialized in constant time in the first step of the procedure due to the additional information
stored with every list of children.
Intuitively, $g$ is how many leaves not belonging to $\leaves(w)$ we still have to see to conclude that indeed 
$\leaves(v_{i})\cap\leaves(u)=\emptyset$ or $\leaves(v_{i})\subseteq \leaves(u)$ for every child $v_i$ of $v$.
We iterate over $\ell\in S$ and access the edge $(v,v_i)$ leading to the subtree
containing $u$ labelled with $\ell$. We decrease $g$ by one and, if $v_i$ is an empty child of $v$
and this is the first time we have seen $v_i$ (in this query) then we add the number of leaves
in the subtree rooted at $v_i$ to $g$. If, after having processed all $\ell\in S$,
$g=0$ then we conclude that $\leaves(u)$ is compatible with $\mathcal{C}$.
The whole process takes $\Oh(|S|\log n)=\Oh(n^{0.5}\log n)$ time.
\end{proof}

Before explaining the details of how to update the additional information, we present the intuition.
Recall that adding $\leaves(u)$ to $\mathcal{C}$ is done by creating a new child $v'$ of $v$
and reconnecting some children of $v$ to $v'$. Let the set of all children of $v$ be $C$
and the set of children that should be reconnected be $C_r$. Note that if $|C_{r}|=1$ or
$|C|=|C_{r}|$ then we do not have to change anything in $T_{c}$. Otherwise,
updating $T_c$ can be implemented using two different methods:
\begin{enumerate}
\item Delete edges from nodes in $C_r$ to $v$. Create a new tree consisting
of a single node $v'$ and make it a child of $v$. Then, make all nodes in $C_r$ children of $v'$.
\item Delete edges from nodes in $C\setminus C_r$ to $v$. Delete the edge from $v$ to its parent $w$.
Create a new tree consisting of a single node $v'$ and make it a child of $w$. Then, make $v$
a child of $v'$ and also make all nodes in $C\setminus C_r$ children of $w$. See Figure~\ref{fig:methods}.
\end{enumerate}
Thus, by using $C_r$ or $C\setminus C_r$, the number of operations can be
either $\Oh(|C_r|)$ or $\Oh(|C|-|C_r|)$. We claim that by choosing the cheaper option
we can guarantee that the total time for modifying the link-cut tree representation of $T_{c}$ is $\Oh(n\log^{2}n)$.
Intuitively, every edge of the final consensus tree participates in $\Oh(\log n)$ operations,
and there are at most $n$ such edges. This is formalized in the following lemma. 

\begin{lemma}
\label{lem:summed}
$\min \{|C_{r}|,|C|-|C_{r}| \}$ summed over all updates of $T_{c}$ is $n\log n$.
\end{lemma}

\begin{proof}
We assume that $2 \leq |C_{r}|<|C|$ in every update, as otherwise there is nothing to change in $T_{c}$.
Then, there are at most $n$ updates, as each of them creates a new inner node and there
are never any nodes with degree 1 in $T_{c}$.

We bound the sum of $\min \{|C_{r}|,|C|-|C_{r}| \}$ by assigning credits to inner nodes of $T_{c}$.
During the execution of the algorithm, a node $u$ with $b$ siblings should have $\log b$
credits. Thus, whenever we create a new inner node we need at most $\log n$ new credits, thus the
total number of allocated credits is $n\log n$. It remains to argue that, whenever we create a new
child $v'$ of $v$ and reconnect some of its children, the original credits of $v$ can be used to
pay for the update and make sure that all children of $v$ and $v'$ have enough credits after the update.

Denoting $x=|C_{r}|$ and $y=|C|-|C_{r}|$, the cost of the update is $\min \{ x,y\}$. The total
number of credits of all children of $v$ before the update is $(x+y)\log(x+y-1)$.
After the update, the number of credits of all children of $v$ is $(y+1)\log y\leq y\log y+\log n$
and the number of credits of all children of $v'$ is $x\log (x+1)$. Ignoring the $\log n$ new credits
allocated to $v'$, the number of available credits is thus:
\[ (x+y)\log(x+y-1) - y\log y - x\log (x-1) = x\log(1+y/(x-1)) + y\log(1+(x-1)/y)\]
which is at least $\min\{x,y\}$ for $x\geq 2$, so enough to pay $\min \{ |C_{r}|,|C|-|C_{r}|\}$
for the update. Hence, the sum is at most $n\log n$.
\end{proof}

\begin{figure}[h]
\includegraphics[width=\textwidth]{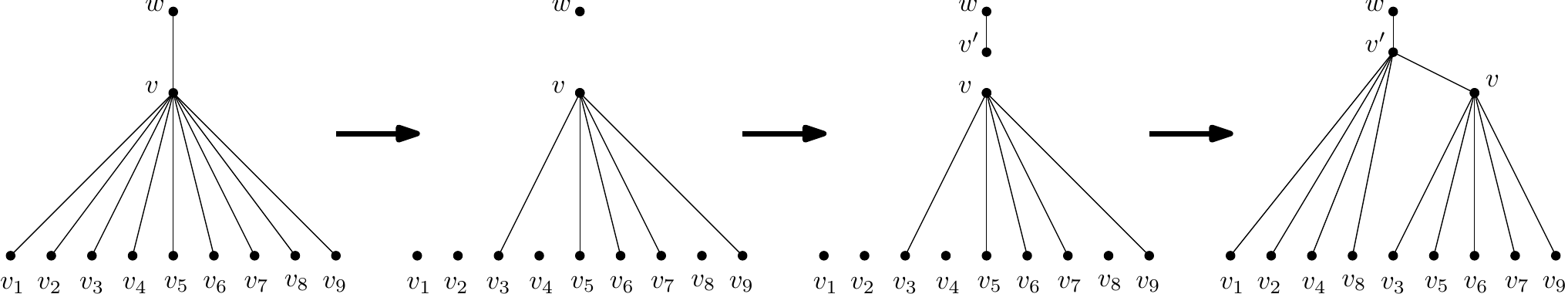}
\caption{Reconnecting children $v_3,v_5,v_6,v_7,v_9$ of $v$ using the second method.\label{fig:methods}
}
\end{figure}

Before presenting the whole update procedure, we need one more technical lemma.

\begin{lemma}
\label{lem:update2}
The procedure for checking if $\leaves(u)$ is compatible with $\mathcal{C}$ can be requested
to return $C_{r}$ in $\Oh(|C_{r}|+n^{0.5})$ time or $C\setminus C_{r}$ in $\Oh(|C|-|C_{r}|+n^{0.5})$ time.
\end{lemma}

\begin{proof}
By inspecting the proof of Lemma~\ref{lem:query}, we see that there are two cases depending on
whether $v$ is a proper ancestor of $\finger(w)$ or not.
\begin{enumerate}
\item If $v$ is a proper ancestor of $\finger(w)$ then $C_r$ can be obtained from $Q$.
More precisely, for every $(v,v_i)\in Q$ we add $v_i$ to $C_{r}$ in $\Oh(|C_{r}|)$ total
time. We can also obtain $C\setminus C_r$ in $\Oh(|C|)=\Oh(|C\setminus C_r|+|S|)=\Oh(|C|-|C_{r}|+n^{0.5})$ time.
\item If $v=\finger(w)$ then, while iterating over $\ell\in S$, if this is the first time we
have seen $v_i$ then we add $v_i$ to $C_{r}$. Additionally, we add all full children of $v$
to $C_{r}$. Thus, $C_{r}$ can be generated in $\Oh(|C_{r}|)$ time.
Similarly, $C\setminus C_{r}$ consists of all empty children of $v$ without the
nodes $v_{i}$ seen when iterating over $\ell\in S$, and so can be generated
in $\Oh(|C\setminus C_r|+|S|)=\Oh(|C|-|C_{r}|+n^{0.5})$ time.
\end{enumerate}
Thus, we can always generate $C_{r}$ in $\Oh(|C_{r}|+n^{0.5})$ time and
$C\setminus C_{r}$ in $\Oh(|C|-|C_{r}|+n^{0.5})$ time. 
\end{proof}

To add $\leaves(u)$ to $\mathcal{C}$, we will need
to iterate over either $C_{r}$ or $C\setminus C_{r}$ (depending on which is smaller).
After paying additional $\Oh(n^{0.5})$ time we can assume that we have access to a list
of the elements in the appropriate set. The additional time sums up to $\Oh(n^{1.5})$,
because there can be only $n$ distinct new sets added to $\mathcal{C}$.

\begin{lemma}
\label{lem:update}
If $\leaves(u)$ is compatible with $\mathcal{C}$ then, after adding $\leaves(u)$ to $\mathcal{C}$
and modifying $T_c$ we can update all additional information in amortized $\Oh(k n^{0.5}\log n)$
time assuming that we add $n$ such sets.
\end{lemma}

\begin{proof}
Recall that $T_{c}$ is maintained using the data structure from Lemma~\ref{lem:lca},
and adding $\leaves(u)$ to $\mathcal{C}$ is implemented by creating a new child $v'$ of $v$
and reconnecting some of the children of $v$ to $v'$. $C$ is the set of all children of $v$ and
$C_r$ is the set of children of $v$ that are reconnected to $v'$. If $|C_r|\leq |C|-|C_r|$ we
iterate over $C_r$ and reconnect them one-by-one. If $|C_r| > |C|-|C_r|$ we iterate over
$C\setminus C_r$ and reconnect them to a new node $w$ that is inserted between $v$ and its
parent. To iterate over either $C_{r}$ or $C\setminus C_{r}$, we extend the query procedure
as explained in Lemma~\ref{lem:update2}. This adds $\Oh(n^{0.5}$ to the time complexity,
but then we can assume that the requested set can be generated in time proportional to
its size. To unify the case of $|C_{r}| \leq |C|-|C_{r}|$ and $|C_{r}| > |C|-|C_{r}|$, we think
that $v$ is replaced with two nodes $v'$ and $v''$, where $v'$ is the parent of $v''$. All nodes
in $C_r$ become children of $v''$ while all nodes of $C\setminus C_r$ become children of $v'$
after iterating over either $C_r$ or $C\setminus C_r$, depending on which set is smaller,
so by Lemma~\ref{lem:summed} in the whole process we iterate over sets of 
total size $n\log n$, so only amortized $\log n$ assuming that we add $n$ sets $\leaves(u)$.

Consider a boundary node $u$. If $\finger(u)\neq v$ then there is no need to update the additional
information concerning $u$. If $\finger(u)=v$ then we need to decide if the finger of $u$ should
be set to $v'$ or $v''$ and update the partition of the children of $\finger(u)$ accordingly.
$\finger(u)$ should be set to $v'$ exactly when, for any $w\in C\setminus C_r$,
$\leaves(w)\cap \leaves(u) = \emptyset$ or, in other words, all nodes in $C\setminus C_r$
are empty with respect to $u$. The groups should be updated as follows:
\begin{enumerate}
\item If $\finger(u)$ is set to $v''$ then we should remove all nodes in $C\setminus C_r$
from the list of empty nodes with respect to $u$ (as they are no longer children of $\finger(u)$).
Other groups remain unchanged.
\item If $\finger(u)$ is set to $v'$ then we should remove all nodes in $C_r$ from the lists.
Additionally, we need to insert $v''$ into the appropriate group: full if all nodes in $C_r$ were
full, empty if all nodes in $C_r$ were empty, and mixed otherwise.
\end{enumerate}
We need to show that all these conditions can be checked by either iterating over the nodes
of $C$ or over the nodes of $C\setminus C_r$, because we want to iterate over the smaller of these.
This then guarantees that the amortized cost of updating the additional information for
a boundary node is only $\Oh(\log n)$, so amortized $\Oh(k n^{0.5}\log n)$ overall.

To check if all nodes in $C\setminus C_r$ are empty with respect to $u$, we can either
iterate over the nodes in $C\setminus C_r$ or iterate over all nodes in $C_r$ and check if all nodes
in $C$ that are full or empty in fact belong to $C_r$ (this is possible because we also keep the total number
of full and empty nodes in $C$). Thus, we can check if $\finger(u)$ should be set to $v'$.

If $\finger(u)$ is set to $v'$ we need to decide where to put $v''$. We only explain how to
decide if all nodes in $C_r$ are full, as the procedure for empty is symmetric. We can either
iterate over all nodes in $C_r$ and check that they are full or iterate over all nodes in $C\setminus C_r$
and check that all nodes in $C$ that are empty or mixed in fact belong to $C\setminus C_r$
(and thus do not belong to $C_r$, so all nodes in $C_r$ are full). Finally, we add the number
of leaves in the subtree rooted at $v''$ (extracted in $\Oh(\log n)$ time) to the
appropriate sum.

It remains to describe how to remove all unnecessary nodes from the lists. Here we do not
worry about having to iterate over the smaller set, because there are only $\Oh(n)$ new edges 
created during the whole execution of the algorithm, so we can afford to explicitly iterate over the nodes
that should be removed, that is, over $C$ or $C\setminus C_r$. For every removed node,
we also subtract the number of leaves in its subtree (extracted in $\Oh(\log n)$ time)
from the appropriate sum. Overall, this adds $\Oh(n\log n)$ per boundary node to the time complexity,
so only amortized $\Oh(k n^{0.5} \log n)$ overall.
\end{proof}

\bibliographystyle{abbrv}
\bibliography{greedy_consensus}

\end{document}